\documentclass{amsart}

\usepackage{amsthm}
\usepackage{latexsym}
\usepackage{amsmath}
\usepackage{amssymb}
\usepackage{algorithm}
\usepackage{algpseudocode}
\usepackage{color}
\usepackage{subfigure}
\usepackage{tikz}
\usetikzlibrary{math, shapes, decorations.pathmorphing, decorations.text}
\tikzmath{
  function doublefactorial(\x) {
    if (\x > 1) then {
      return \x * doublefactorial(\x-2);
    } else {
      return 1;
    };
  };
}

\allowdisplaybreaks

\newcommand{\old}[1]{}

\newtheorem{theorem}{Theorem}[section]
\newtheorem{lemma}[theorem]{Lemma}
\newtheorem{corollary}[theorem]{Corollary}

{
\theoremstyle{definition}
\newtheorem{definition}[theorem]{Definition}
\newtheorem{remark}[theorem]{Remark}
}

\newcommand{\R}{\mathbb{R}}

\usepackage{hyperref,xcolor}
\hypersetup{
     colorlinks   = true,
     urlcolor    = teal,
	 citecolor = teal,
	 linkcolor = teal
}

\title[A Partisan Districting Protocol with Provably Nonpartisan Outcomes]{A Partisan Districting Protocol\\
with Provably Nonpartisan Outcomes}

\author{Wesley Pegden}
\address{Department of Mathematical Sciences, Carnegie Mellon University}
\email{wes@math.cmu.edu}

\author{Ariel D. Procaccia}
\address{Computer Science Department, Carnegie Mellon University}
\email{arielpro@cs.cmu.edu}

\author{Dingli Yu}
\address{Institute for Interdisciplinary Information Sciences, Tsinghua University}
\email{yudingli\_good@126.com}

\begin{document}

\begin{abstract}
We design and analyze a protocol for dividing a state into districts, where parties take turns proposing a division, and freezing a district from the other party's proposed division.   We show that our protocol has predictable and provable guarantees for both the number of districts in which each party has a majority of supporters, and the extent to which either party has the power to pack a specific population into a single district.  
\end{abstract}

\maketitle

\section{Introduction}
In the American system of democracy, local representatives from a state are elected to the national House of Representatives in direct local elections, held in districts of (roughly) equal population. These Congressional districts of a state are redrawn every 10 years, following the decennial census.

Because the Congressional representatives of a state are elected in local district elections, there is no guarantee that the political makeup of the state's elected slate of representatives will mirror the political composition of statewide vote casts in the election.  In practice, this seemingly desirable target is missed by a wide mark; for example, in Pennsylvania in 2012, 51\% of the population voted for Democrat representatives, yet Republican representatives won 13 out of 18 district elections.

This phenomenon, popularly known as \emph{gerrymandering} --- where political parties carefully draw districts to maximize their advantage in election outcomes --- has been studied on several fronts.  Researchers have worked to propose simple tests to quantify gerrymandering with greater granularity than the election outcomes provide \cite{wang,mcdonald}, and to distinguish gerrymandering from natural artifacts of political geography \cite{chenrodden,duke,outliers}.  

In some states, there has been an earnest effort to reduce political influence over the districting process through the establishment of independent redistricting commissions.  For example, in Arizona, redistricting is currently carried out by a commission composed of 10 Democrats, 10 Republicans  and 5 Independents.

Our goal in this paper is to propose and analyze a protocol for fair redistricting that can be carried out without ``independent'' agents; in particular, a redistricting commission using our protocol can be composed of an even number of members, drawn from the two major political parties of the state.  Districting with our protocol would enable reasonable districts to be drawn for a state without requiring effective mechanisms identifying truly independent agents.

Motivated by classical cake-cutting problems with ``I-cut-you-choose'' types of solutions~\cite{BT96,Pro13}, our algorithm leverages competition between two political entities to create a reasonable districting in a turn-based protocol.  We will prove that our protocol has desirable properties in idealized settings.

\section{Setting and results}
We consider the districting problem as a competition between two parties.  A \emph{state} will be modeled as a continuously divisible object, some subset of which is loyal to Player 1, and some subset of which is loyal to Player 2.  We will sometimes ignore geometry and sometimes pay attention to it, so we really have two models of a state.

In the first model, the state is an interval $[0,n]$ ($n$ is the number of districts) and an $n$-\emph{districting} is a a collection of $n$ disjoint unit-measure subsets of the interval.  In this model, the measure of a subset corresponds to a population, where 1 unit of population is the size of 1 district.  In the second model, the state is a subset $X\subset \R^2$ topologically equivalent to a disc, together with a population density $\phi:X\to [0,1]$.  We are given that $\int_{X} \phi=n$ and an $n$-districting of $X$ is a division of $X$ into $n$ disjoint \emph{connected} subsets $X_i$, each satisfying that $\int_{X_i} \phi=1$.

Given a district in a districting of a state, the district is \emph{loyal} to the player with a majority of loyalty in the district.  (In cases of exact ties, we break them arbitrarily, say, to Player 1).   The outcome for a player from the districting is the player's \emph{slate}, which is the number of districts loyal to him.

We now turn to a discussion of natural districting protocols, culminating with our own.

\subsection{One player decides}
In this trivial protocol, some external process chooses one player, and that player freely chooses the districting.  Subject to moderate legal oversight, this is essentially the protocol currently used in most states.  In practice, the external process which chooses the favored player is often control of the state legislature (which itself is influenced by gerrymandering of state-level districts), or, as in Pennsylvania, control of the state supreme court.

\subsection{Independent agent protocols}
If benevolent independent agents are available in the problem model, then a natural solution is to simply allow independent agents to draw the districting.  In spirit, this is the approach of redistricting commissions such as Arizona's.  Our goal is to eliminate the trust required of independent agents.

Interesting work by  Landau, Reid, and Yershov \cite{LRY} and Landau and Su \cite{LandauSu} developed protocols with a moderate dependence on an independent agent.  Essentially, an independent agent is used just to choose a suitable division of the state into two parts, and assign one part (and a target number of districts) to each player.  Each player then freely districts his part, and the result is combined into a districting of the state.  They proved an elegant theoretical guarantee for their protocol under optimum play: Each player will achieve at least the average of their best possible slate and their worst possible slate of representatives, among districtings respecting the split-line chosen by the independent agent.

Apart from the dependence on an independent agent, their protocol has one other feature which may dissuade its adoption in practice.  Since the outcome of the protocol is a districting in which each player freely chose districtings on their side of the split, the intended result is a districting of a state in which one side is gerrymandered for Player 1, while the other is gerrymandered for Player 2.  Although this produces a reasonable outcome in terms of the slates won by each player, which is our primary outcome of interest, such a protocol will not reduce (and in some cases, may even exacerbate) other maligned effects sometimes attributed to gerrymandering, such as entrenchment of incumbent representatives and the rise of political extremism  \cite{lindgren2013effect,carson2014reevaluating}. Similar problems plague the I-cut-I-freeze protocol, below.  We will give a formalization capturing this phenomenon (Definition \ref{d.target}) and prove that our protocol avoids it (Theorem \ref{t.target}).

\subsection{I-cut-I-freeze}
A very simple multiround districting protocol is to simply have the players take turns adding districts with the correct population size to an initially empty districting, until the districting is complete.  The following version of this ensures that completion is always possible.
\subsubsection*{I-cut-I-freeze:}   Initialize $n$ to be the number of districts the state is to be divided into.  Initially no districts are frozen.  On each player's turn (while $n>0$), the player:
\begin{enumerate}
\item Redistricts the unfrozen part of the state into $n$ districts;
\item Chooses one of the new districts to be frozen;
\item Updates $n\leftarrow (n-1)$.
\end{enumerate}

One attractive theoretical feature of this protocol is its  slate guarantee (ignoring any constraints due to geography); the straightforward proof is omitted.

\begin{theorem}
In the geometry-free model, if more than $\ell/ n$ of the state is loyal to Player $i$, then Player $i$ can achieve a slate of size $\ell$ from a $n$-districting, in the I-cut-I-freeze protocol.
\end{theorem}

On the other hand, like the independent agent protocols discussed above, the I-cut-I-freeze protocol allows each player to freely draw many districts.  In particular, each political party will be able to pack minority populations, reinforce their incumbents, etc.  Although from the standpoint of the final slates for each player, competition balances the two sides, these features may make this protocol undesirable nonetheless.

To rigorously capture the power afforded each player by the I-cut-I-freeze protocol and the related independent agent protocols, we define the following property of a districting protocol:
\begin{definition}
\label{d.target}
We say that a $n$-districting protocol has the $B$-target property if for any $i\in \{1,2\}$ and any target subset of the state of measure/cardinality $\frac{1}{n}$ of the state's total, Player $i$ has a strategy to ensure that at most a $1/B$ fraction of the target intersects any single district.
\end{definition}
Note that trivially any protocol has the $1$-target property.  Moreover, $1$ is the largest value of $B$ for which the I-cut-I-freeze protocol has the $B$-target property, since, for example, on Player 1's first turn, he can always create a district equal to any choice for the target.  We will see that the situation for the our proposed protocol is very different.

\subsection{I-cut-you-freeze}
The I-cut-you-freeze protocol is the main subject of this paper.  Essentially, the motivation is to reduce the influence a single player can exert unilaterally on the drawing of any single district.  Although each player will still draw $n/2$ districts (up to rounding) that are in the final districting, they will no longer have control over which of the districts they draw in the course of the protocol end up in the final districting.

\subsubsection*{I-cut-you-freeze:}   Initialize $n$ to be the number of districts the state is to be divided into.  At the beginning of the protocol, Player 1 gives a districting of the state into $n$ districts, and passes it to Player 2. On each player's subsequent turn (while $n>0$), the player who has just been given the districting:
\begin{enumerate}
\item Chooses an unfrozen district to be frozen (in the districting received from the other player);
\item Updates $n\leftarrow (n-1)$.
\item Redistricts the still unfrozen part of the state into $n$ districts, and passes the new districting back to the other player.
\end{enumerate}
At the end of a protocol, we have a districting in which half of the districts were drawn by Player 1 and frozen by Player 2, and vice versa.

Let us define $\sigma(n,s)$ to be the slate which will be won by Player $1$ under optimum play of the I-cut-you-freeze protocol for $n$-districtings, in the setting where the state is modeled as an interval of length $n$, when the subset of the state loyal to Player $i$ has measure $s$.  We characterize the outcome of our protocol asymptotically as follows:

  \begin{theorem}
\label{t.a-seats}  We have
  \[
  \lim_{n\to \infty} \frac{\sigma(n,\alpha n)}{n}=
  \begin{cases}
    2\alpha^2 & \text{for } \alpha\leq \tfrac 1 2\\
    1-2(1-\alpha)^2 & \text{for }\alpha>\tfrac 1 2.
  \end{cases}
  \]
  \end{theorem}

  \noindent We can also explicitly characterize the small-$n$ behavior in all cases.

  \begin{theorem}
    We have $\sigma(n,s)\geq k$ if and only if
    \begin{equation}
      s\succeq
      \begin{cases}
\displaystyle    \frac{(n-1)!!(2k+[2\nmid n]-2)!!}{2(2k+[2\nmid n]-3)!!(n-2)!!} & \text{for }k\leq n/2\\[15pt]
\displaystyle    n-\frac{n!!(2(n-k)-[2\nmid n]+1)!!}{2(2(n-k)-[2\nmid n])!!(n-1)!!} & \text{for }k>n/2.
      \end{cases}
    \end{equation}
    Here $\succeq$ is interpreted as $\geq$ if ties are broken in favor of Player 1, and $>$ if they are broken in favor of Player 2;  $[2\nmid n]$ is $1$ or $0$ according to the parity of $n$; and $n!!$ is the double factorial $n(n-2)(n-4)\cdots(2\text{ or }1)$.
    \label{t.seats}
  \end{theorem}
  By contrast, for the trivial One-player-decides algorithm (with Player 1 deciding), we would have $\sigma(n,s)\geq k$ if and only if $s\geq k/2$; Figure~\ref{f.curve} illustrates a comparison.
\begin{figure}[t]
\begin{center}
\begin{tikzpicture}[xscale=3,yscale=3]

      \draw[-|] (0,0) -- (1,0) node[right] {$\alpha$};
      \draw (1,0) node[below] {\tiny $1$};
      \draw (0,0) node[below] {\tiny $0$};

\draw[-|] (0,0) -- (0,1) node[above] {$\sigma_{10}$};
      \draw (0,1) node[left] {\tiny $10$};
      \draw (0,0) node[left] {\tiny$0$};

\draw (0,0) -- (.05,0) -- (.05,.1) -- (.1,.1) -- (.1,.2) -- (.15,.2) -- (.15,.3) -- (.2,.3) -- (.2,.4) -- (.25,.4) -- (.25, .5) -- (.3,.5) -- (.3,.6) -- (.35,.6) -- (.35,.7) -- (.4,.7) -- (.4,.8) -- (.45,.8) -- (.45, .9) -- (.5,.9) -- (.5,1) -- (1,1);

\draw [thick](0,0) -- (.123047,0) -- (.123047,.1) -- (.245094,.1) -- (.245094,.2) -- (.328125,.2) -- (.328125,.3) -- (.39375,.3)  -- (.39375,.4) -- (.45,.4) -- (.45,.5) -- (.5,.5) -- (.5,.6) -- (.555556,.6) -- (.555556,.7) -- (.619048,.7) -- (.619048,.8) -- (.695238,.8) -- (.695238,.9) -- (.796825,.9) -- (.796825,1) -- (1,1);
\fill (0.5,0.5) ellipse(.0175cm and .0175cm);
    \end{tikzpicture}
    \hspace{1cm}
    \begin{tikzpicture}[xscale=3,yscale=3]

\draw[-|] (0,0) -- (1,0) node[right] {$\alpha$};
      \draw (1,0) node[below] {\tiny$1$};
      \draw (0,0) node[below] {\tiny$0$};

\draw[-|] (0,0) -- (0,1) node[above] {$\sigma_n$};
      \draw (0,1) node[left] {\tiny$n$};
      \draw (0,0) node[left] {\tiny$0$};

\draw (0,0) -- (.5,1) -- (1,1);

\draw[thick,scale=1,domain=0:.5,smooth,variable=\x,black] plot (\x,{2*\x * \x});
\draw[thick,scale=1,domain=.5:1,smooth,variable=\x,black] plot (\x,{1-(2*(1-\x) *(1-\x))});
\fill (0.5,0.5) ellipse(.0175cm and .0175cm);
\end{tikzpicture}
  \end{center}
\caption{\label{f.curve} The proportion of seats $\sigma_n$ won by Player 1 under optimum play is plotted against his loyalty $\alpha$ of the state.  Bold lines are the result of the I-cut-you-freeze algorithm, while light lines are the result of the trivial ``One player decides'' algorithm (with Player 1 deciding).  On the left is the case $n=10$, while the right gives the curves in the large-$n$ limit.  The point $\alpha=\tfrac 1 2$, $\sigma=\tfrac n 2$ is marked on both curves.  By Corollary \ref{c.half}, this point lies on the $\sigma_n(\alpha)$ curve for our protocol for every value of $n$.}
\end{figure}
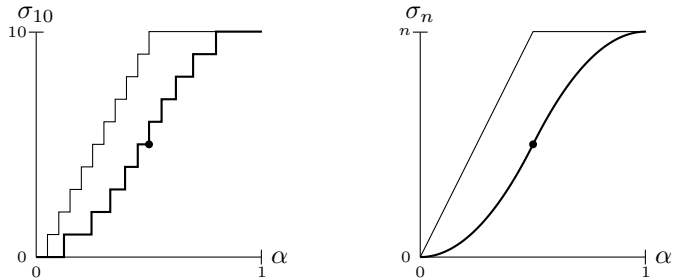

The following corollary is immediate from Theorem \ref{t.seats}. It implies, in particular, that majority shares are perfectly aligned, up to rounding.

  \begin{corollary}\label{c.half}
Player 1 wins more than $\lfloor n/2 \rfloor$ seats if and only if his loyal subset has measure $s\succeq n/2$.
  \end{corollary}

As mentioned above, we also wish to capture the diminished power either player has over any single district under our method compared with the previously discussed independent agent protocols, and the I-cut-I-freeze protocol.  It is possible to prove this rigorously even in a geometric setting where districts are required to connected, as we do in Section \ref{s.geometric}.

\begin{theorem}\label{t.target}
The I-cut-you-freeze protocol for a $n$-districting has the $B$-target property for $B=\frac{\sqrt{n}}{2}$, in both the geometric and non-geometric models.
\end{theorem}

\begin{remark}
It is natural to wonder why our analysis of slates (Theorems \ref{t.seats} and \ref{t.a-seats}) is carried out in a geometry-free setting, while our analysis of the $B$-target property allows geometry.  The problem is that in full generality, a setting respecting geometry can defy analysis of slate outcomes in a way that, we suspect, would not actually correspond well to the real-world behavior of our protocol.  For example, suppose we model a state as a topological unit disc $X$, as above, which has not only a population density $\phi:X\to [0,1]$ but also densities $\phi_A$ and $\phi_B$ of those loyal to Players A and B, where $\phi=\phi_A+\phi_B$.  The outcome of our protocol will now be highly dependent on the geometric relationship between $\phi_A$ and $\phi_B$; for example, if we allow that $\phi_A>\phi_B$ everywhere, then Player A will win every district no matter what choices the players make, and, in principle, this requires only that $s>n/2$. By contrast, in an application to redistricting in the United States, such a consistent relationship between the player loyalties does not occur.  We discuss the applicability of our idealized settings to real-world implementation of our protocol in Section \ref{s.realworld}, but roughly speaking, we believe that Theorems \ref{t.a-seats} and \ref{t.seats} do capture concrete advantages of our protocol over the One-player-decides protocol, which we would expect to persist in the real world; namely, that the protocol is nearly symmetric with respect to which player gets the first move, and produces generally reasonable outcomes.
\end{remark}

\section{Slates from optimal play}
In this section we analyze the I-cut-you-freeze protocol in terms of the relation between the measure of the subset loyal to a player and his slate, namely Theorems \ref{t.a-seats} and \ref{t.seats}.

Let us denote the measure of the subset of $[0,n]$ loyal to Player 1 by $s_1^n$, and the measure of the subset of $[0,n]$ loyal to Player 2 by $s_2^n$. Formally, the game can be expressed recursively as the following procedure for a given position, which, with $k=n$, $s_1=s_1^n$ and $A=1$, returns the slate of Player 1.  For $k\in \mathbb{N}^+$, $s_1\in [0,k]$, and $A\in \{1,2\}$, we define:

\medskip
\newcommand{\game}{\mathrm{\textsc{Game}}_{1}}
\begin{algorithmic}[1]
    \Procedure{Game$_1$}{$k,s_1,A$}\Comment{Player $A$ divides first}
    \State Player $A$ chooses $k$ numbers in $[0,1]$: $x_{k,1},\ldots,x_{k,k}$, such that \[\sum_{i=1}^k x_{k,i} = s_1\]
    \State Player $B$ chooses an integer $i_k\in [k]$, where $\{A,B\}=\{1,2\}$
    \State \textbf{return} $\game(k - 1, s_1-x_{k,i_k},B)$ + $[x_{k,i_k}\geq\tfrac 1 2]$
   \EndProcedure
\end{algorithmic}
\medskip

We also set $\game(0,0,A)=0$.  Now let
\[
\game(n,s_1,1),\,\game(n-1,s_1-x_{n,i_n},2),\dots
\]
be the procedures (game positions) encountered recursively in the course of the game $\game(n,s_1,1)$.  We call steps 2 and 3 of each such procedure a \emph{round} of the game $\game(n,s_1,1)$, and number the rounds in reverse, beginning with round $n$ and ending with round 1.  In particular, round $t$ begins with some player choosing $t$ numbers $x_{t,1},\dots,x_{t,t}$.

We let $s_1^{(t)}$ and $s_2^{(t)}$ be the fraction of the unfrozen state loyal to Players 1 and 2, respectively, at the beginning of round $t$.  In particular, we can set
\begin{align*}
  s_1^{(t)}&=s_1-\sum_{k=t+1}^{n} x_{k,i_k}\\
  s_2^{(t)}&=t-s_1^{(t)}.
\end{align*}

Let $f(k,s_1,A)$ be the output of $\game(k,s_1,A)$ when the two players play optimally. (Note that this function always returns the slate of Player 1, even if $A=2$.)  Then we have that
\begin{align}
\label{f1}  f(k,s_1,1) &= \max_{x_{k,1},\ldots,x_{k,k}} \min_{i\in [k]} \left(f(k -1,s_1 - x_{k,i},2) + [x_{k,i}\geq\tfrac 1 2]\right),\quad\text{and}\\
\label{f2}  f(k,s_1,2) &= \min_{x_{k,1},\ldots,x_{k,k}} \max_{i\in [k]} \left(f(k -1,s_1 - x_{k,i},1) + [x_{k,i}\geq\tfrac 1 2]\right).
\end{align}

It is intuitively obvious that $f(k,s_1,A)$ should be monotone with respect to $s_1$, and indeed, this follows immediately from induction on \eqref{f1} and \eqref{f2}:
\begin{lemma}\label{lemma:mono}
  $f(k,s_1,A)\leq f(k,s_1',A)$ if $s_1< s_1'$.\qed
\end{lemma}

\noindent The following lemma shows that in optimum play, players will divide the unfrozen region into districts with at most two distinct loyalty values.
\begin{lemma}\label{lemma:two-value}
  For any $\game(k,s_1,A)$, there are numbers $\omega\geq \tfrac 1 2>\lambda$ such that under optimum play, Player $A$ will choose each $x_{k,i}$ to be $\omega$ or $\lambda$.
\end{lemma}
\begin{proof}
   First consider $A = 1$, let $W = \{i:x_{k,i}\geq \tfrac 1 2\}$ and $L = [k] - W$.  By Lemma \ref{lemma:mono} applied to $\game(k-1,s_1-x_{k,i},2)$ which is encountered at the end of the round, Player 2's optimum move is to either choose $i=\arg\max_{j\in W} x_{k,j}$ or $i=\arg\max_{j\in L} x_{k,j}$. So under optimum play, Player 1 will assign an identical number $\omega$ to districts in $W$, and an identical number $\lambda$ to districts in $L$, where $\omega = \sum_{i\in W} x_{k,i} / |W|$ if $W\neq \emptyset$ and $\lambda = \sum_{i\in L} x_{k,i} / |L|$ if $L\neq \emptyset$.
\old{   Therefore,
   \[f(n,s_1,1)=\max\left\{\begin{array}{l}
   \max_{a,b}\{\min\{f(n-1,s_1-a,2) + 1, f(n-1,s_1-b,2)\}\}, \\
   f(n-1,s_1-s_1/k,2)+[s_1/k \geq \tfrac 1 2]\end{array}\right\},\]
   where $a \geq \tfrac 1 2> b$ and $\exists m \in [n -1], am + b(n - m) = s_0$.
}
   The case of $A=2$ is analogous.
\end{proof}

In round $t$, if $s_1^{(t)} \geq t/2$, we say Player 1 is stronger and Player 2 is weaker in this round; otherwise, Player 2 is stronger and Player 1 is weaker. (Note that, in general, which player is strongest may change from round to round, even under optimum play.)

Lemma \ref{lemma:two-value} implies that Player $A$'s move is thus completely characterized by the choice made for $\omega$ and $\lambda$.  Assuming without loss of generality that $A=1$, we see that Player 2's response will result either in the value $f(k-1,s_1-\lambda,2)$ or the value $f(k-1,s_1-\omega,2)+1$, assuming the remaining gameplay is optimal.  In particular, we have the following:
\begin{lemma}\label{lemma:best-pairs}
  Given possible choices $(\omega, \lambda)$ and $(\omega',\lambda')$ for Player $A$ satisfying $\omega \leq \omega'$ and $\lambda \leq \lambda'$, the choice $(\omega,\lambda)$ dominates the choice $(\omega',\lambda')$ if $A = 1$; otherwise, the choice $(\omega',\lambda')$ dominates the choice $(\omega,\lambda)$.
\end{lemma}
\noindent This immediately implies the following:
\begin{lemma}\label{lemma:stronger-strategy}
  In any $\game(k,s_1,A)$, Player A, if stronger, chooses $x_{k,1}=x_{k,2}=\dots=x_{k,k}$ in optimum play.  In particular:
  \[
    A\text{ stronger}\implies f(k,s_1,A) = f(k-1,s_1 - s_1 / k, B) + 1.
  \]
\end{lemma}
\old{
\begin{proof}
  We prove the case where $A = 1$; the case where $A=2$ is analogous. We let $\gamma=\max_i x_{k,i}$ be the largest $x_{k,i}$ chosen by Player A.  Assuming the lemma fails, we have $x_{k,i}>s_1/k$.
  Therefore, monotonicity of $f$ (Lemma \ref{lemma:mono}) would give that
\[
f(k-1,s_1 - s_1 / k, 2) + 1 > f(k-1,s_1-\gamma,2) + 1,
\]
right hand side of this inequality is a game value achievable by Player 2 assuming the lemma fails, and the left hand side is the only value achievable if Player 1 plays as claims.
\end{proof}
}

The following lemma will serve as a base case, if you will, for our larger analysis, characterizing the outcome of play once the two players are roughly even.
\begin{lemma}\label{lemma:close-case}
  If $(k-1)/2 \leq s_1 < k/2$, \[f(k,s_1,1)=\lfloor k / 2 \rfloor;\]
  and if $k/2 \leq s_1 < (k+1)/2$, \[f(k,s_1,2)=\lceil k / 2 \rceil.\]
\end{lemma}

\begin{proof}
  Assume that $(k-1)/2 \leq s_1 < k/2$; the other case is analogous. We prove the lemma by induction on $n$. When $k \leq 2$, it is trivial. When $k > 2$,
\begin{equation}
\label{eq:maxmin}
f(k,s_1,1)=\max\left\{\begin{array}{l}
   \max_{\omega,\lambda}\{\min\{f(k-1,s_1-\omega,2) + 1, f(k-1,s_1-\lambda,2)\}\}, \\
   f(k-1,s_1-s_1/k,2)+[s_1/k \geq \tfrac 1 2]\end{array}\right\},
\end{equation}
  where the max over $\omega,\lambda$ is taken over pairs satisfying $\omega \geq \tfrac 1 2 > \lambda$, with the property that there exists $m\in [k-1]$ such that $\omega m+\lambda (k-m) = s_1$. Here we can assume $m = k-1$, otherwise, one can easily find $\omega'$ and $\lambda'$ such that $\omega > \omega' \geq \tfrac 1 2 > \lambda > \lambda'$ and $\omega'(k-1)+\lambda' = s_1$. It remains to analyze each of the expressions in Equation~\eqref{eq:maxmin}.

  First, since
$$s_1-\omega < \frac k 2 - \frac 1 2 = \frac{k-1}{2},$$
we have (by Lemma~\ref{lemma:stronger-strategy})
\[f(k-1,s_1-\omega,2)+1 = f\left(k-2, (s_1 - \omega)\frac{k-2}{k-1}, 1\right) + 1.\]
  Then, since
\begin{align*}
(s_1 - \omega)\frac{k-2}{k-1} \geq s_1\frac{(k-2)^2}{(k-1)^2} \geq \frac{1}{2}\cdot \frac{(k-2)^2}{k-1} > \frac{k-3}{2},
\end{align*}
and
$$
(s_1 - \omega)\frac{k-2}{k-1}<\frac{k-1}{2}\cdot\frac{k-2}{k-1}=\frac{k-2}{2},
$$
we have by the induction assumption that
\[f\left(k-2, (s_1 - \omega)\frac{k-2}{k-1}, 1\right) = \left\lfloor \frac{k-2}{2} \right\rfloor.\]

Second, $$\frac k 2 > s_1 \geq s_1 -\lambda = \omega(k-1) \geq \frac{k-1}{2},$$ so, by the induction assumption, $f(k-1,s_1-\lambda,2) = \lceil (k-1) / 2 \rceil$.

Third, for $f(k-1,s_1-s_1/k,2)$, since
$$s_1 - \frac{s_1}{k}  = s_1\frac{k-1}{k} < \frac{k-1}{2}$$ and $s_1 (k-2)/k\geq (k-3)/2$, we have \[f\left(k-1,s_1-\frac{s_1}{k},2\right) = f\left(k-2, s_1 \frac{k-2}{k}, 1\right) = \left\lfloor \frac{k-2}{2}\right\rfloor.\]

By Equation~\eqref{eq:maxmin}, we conclude that
$$f(k,s_1,1)= \max\{\min\{\lfloor (k-2) / 2 \rfloor + 1, \lceil (k-1) / 2 \rceil\},\lfloor (k-2) / 2 \rfloor\} = \lfloor k / 2 \rfloor.$$
\end{proof}

\begin{lemma}\label{lemma:weaker-strategy}
   In any $\game(k,s_1,A)$, if Player $A$ is weaker and $s_A \succeq 1/2$, the following strategies for players are optimal:
  \begin{itemize}
    \item Let $m = \lfloor 2s_A \rfloor$ if $A =1$, $m = \lceil 2s_A \rceil - 1$ if $A = 2$. Player $A$ will divide the resources such that his proportion in each district is either 0 or $s_A/m$.
    \item Player $B$ will choose a district where his loyalty is 1.
  \end{itemize}
  In particular:  if $s_1 < k/2$, \[f(k,s_1,1) = f(k-1,s_1,2);\]
  if $s_1 \geq k/2$, \[f(k,s_1,2) = f(k-1,s_1-1,1) + 1.\]
\end{lemma}
\begin{proof}

  We write this proof for the case $A = 1$; the case $A=2$ is similar. For $s_1 \geq (k-1) / 2$, it is obvious from the proof of Lemma \ref{lemma:close-case}.

For $s_1 < (k-1)/2$, prove the lemma by induction on $k$. The base case is $k\leq 2$; we have $s_1<1/2$, so Player $1$ cannot win any districts, and the claim on $f$ is obvious.

For $k > 2$, we first prove that if Player 2's optimal strategy always chooses a piece where Player 2's loyalty is greater than $1/2$, then Player 1's optimal strategy is as above. Indeed, this implication is clear, since given that Player 2 will choose a piece where his loyalty is greater than $1/2$, Player 1 wishes this piece to have maximum loyalty to Player 2, by monotonicity.  In particular, his strategy will only produce districts which have loyalty 1 or less than $1/2$ to Player 2; equivalently, $0$ or more than $1/2$ to himself.  (The condition $s_A\succeq 1/2$ ensures that this is always possible.)
  \old{
    Under this assumption on Player 2's optimal strategy, we have
  \begin{equation}
    f(k-1,s_1-s_1 / m,2) + 1 > f(k-1,s_1,2) \geq  f(k-1,s_1 - b,2)
  \end{equation}
  for any $b>0$, where this second inequality is a consequence of monotonicity (Lemma \ref{lemma:mono}).
    Therefore,
  \begin{equation}
    f(k-1,s_1,2) \geq \max_{a,b}\{\min\{f(k-1,s_1-a,2) + 1, f(k-1,s_1-b,2)\}\},
  \end{equation}
  and
  \begin{equation}
    f(k-1,s_1,2) \geq f(k-1,s_1 - s_1 / k, 2).
  \end{equation}}
  This gives that $f(k,s_1,1) = f(k-1,s_1,2)$, as desired.

 Thus it suffices to prove that Player 2's optimal strategy will always choose a district where his loyalty is more than $1/2$.  We prove that this is the case in response to any optimal move by Player 1 of the form guaranteed to exist by Lemma \ref{lemma:two-value}.  By Lemma \ref{lemma:best-pairs}, we may assume that Player 1 makes a minimal choice of the pair $(\omega,\lambda)$.  In particular, for any $\lambda\geq 0$, there is some minimum $\omega_\lambda$ for which the pair $(\omega_\lambda,\lambda)$ is feasible, and an interval $I=[0,\beta]$ such that $\omega_\lambda$ is decreasing with respect to $\lambda$ for $\lambda\in I$ and and such that $\omega_\beta=1/2$.  In particular, it suffices to prove that for any move by Player 1 of the form $(\omega_\lambda,\lambda)$ for $\lambda\in I$, Player 2's optimal response will be to choose a district where his loyalty is more than $1/2$.  Finally, monotonicity implies that it suffices to consider the case of the pair $(\omega_0,0)$ since Player 2's outcome from choosing district with minority loyalty worsens, and his outcome from choosing a majority district improves, as $\lambda$ increases. Of course, $\omega_0$ is precisely $s_A/m$, so we may, in fact, assume that Player 1 employs the strategy that this lemma claims is optimal.

 Now, for the sake of a contradiction, suppose that $s_1^{(k)}< (k-1)/2$ and Player 2 is \emph{strictly} better off choosing a district where he has minority loyalty. Letting $m = \lfloor 2s_1^{(k)} \rfloor$, we have that
 \[s_1^{(k-2)} = s_1^{(k-1)} \frac{k-2}{k-1} =  s_1^{(k)} \frac{m-1}{m}\cdot \frac{k-2}{k-1} < \frac{k-2}{2},\]
 where the first equality follows from Lemma \ref{lemma:stronger-strategy}.
 By the induction assumption, in round $k-2$, the players follow the claimed optimal strategies, so \[s_1^{(k-3)} = s_1^{(k-2)} = \frac{s_1^{(k)}(m-1)(k-2)}{m(k-1)}.\]
 Now we compare this result with the case where Player 2 chooses a majority-loyal district in round $k$ and a minority-loyal district in round $k-2$. It holds that
 \[s_1^{(k-2)} = s_1^{(k-1)} \frac{k-2}{k-1} = s_1^{(k)} \frac{k-2}{k-1} < \frac{k-2}{2}. \]
  Let $$m' = \lfloor 2s_1^{(k - 2)} \rfloor = \left\lfloor s_1^{(k)} \frac{k-2}{k-1}\right \rfloor,$$ then \[s_1^{(k-3)} = s_1^{(k-2)} \cdot \frac{m'-1}{m'} = \frac{s_1^{(k)}(m'-1)(k-2)}{m'(k-1)}.\]
  Since $m' \leq m$, $$\frac{s_1^{(k)}(m'-1)(k-2)}{m'(k-1)} \leq \frac{s_1^{(k)}(m-1)(k-2)}{m(k-1)}.$$
Since in both cases Player 2 wins the same number of districts, this is a contradiction to our assumption that it is strictly optimal for Player 2 to choose a district where he has minority loyalty.
\end{proof}

\begin{proof}[Proof of Theorems \ref{t.a-seats} and \ref{t.seats}]
  Combining Lemma \ref{lemma:stronger-strategy}, Lemma \ref{lemma:close-case} and Lemma \ref{lemma:weaker-strategy}, the recurrence relation of $f$ is solved.

  One can easily find the following: for $\kappa = n,n-2,n-4,\ldots$, $f(n,s_1,1) \geq \lfloor \kappa / 2\rfloor$ if and only if \[s_1 \geq \frac{(n-1)!!(\kappa-2)!!}{2(\kappa-3)!!(n-2)!!};\]
  and $f(n,s_1,2)\leq n-\lfloor \kappa / 2\rfloor$ if and only if \[s_1 < n - \frac{(n-1)!!(\kappa-2)!!}{2(\kappa-3)!!(n-2)!!}.\]
  For $\kappa = n-1,n-3,n-5,\ldots$, $f(n,s_1,1) \leq n-\lfloor \kappa / 2\rfloor$ if and only if \[s_1 < n - \frac{n!!(\kappa-2)!!}{2(\kappa-3)!!(n-1)!!};\]
  and $f(n,s_1,2) \geq \lfloor \kappa / 2\rfloor$ if and only if \[s_1 \geq \frac{n!!(\kappa-2)!!}{2(\kappa-3)!!(n-1)!!}.\]
  Writing $k=\lfloor \kappa/2\rfloor$ or $k = n - \lfloor\kappa / 2\rfloor+1$ gives Theorem \ref{t.seats} as stated.  Theorem \ref{t.a-seats} then follows by Stirling's approximation.
\end{proof}

\section{The $B$-target property}
In this section we prove Theorem \ref{t.target}.  We begin by proving the theorem in the nongeometric setting, which is simpler.
\subsection{Nongeometric setting}
In what follows, we suppose that Player 2 wants to ensure that as small as possible fraction of the target intersects any single district, and Player 1 opposes this goal.  Our analysis does not depend on which of these players has the first move, however.

Let the measure of the target subset be $s_T$. Then the game can be expressed as the following recursive procedure, whose value for $k=n, s=s_T$ is precisely the maximum measure of the target that intersects any single district, under optimum play.

\medskip

\newcommand{\gamet}{\mathrm{\textsc{Game}}_{2}}
\begin{algorithmic}[1]
    \Procedure{Game$_2$}{$k,s,A$}\Comment{Player $A$ divides first}
    \State Player $A$ chooses $k$ numbers in $[0,1]$: $x_{k,i},\ldots,x_{k,k}$, such that \[\sum_{i=1}^n x_{k,i} = s\]
    \State Player $B$ chooses an integer $i\in [k]$, where $\{A,B\}=\{1,2\}$
    \State \textbf{return} $\max\{\gamet(k -1, s-x_{k,i},B), x_{k,i}\}$
   \EndProcedure
\end{algorithmic}

\medskip

\begin{theorem}
  Under optimum play, $\gamet(n,s_T,2)$ has value $\frac{s_T(n-1)!!}{n!!}$, while $\gamet(n,s_T,1)$ has value $\frac{s_T(n-2)!!}{(n-1)!!}$. In particular, the I-cut-you-freeze protocol for an $n$-districting has the $B$-target property in the geometry-free setting for
  \begin{equation}
    \label{eq:B}
  B = \min\left\{\frac{n!!}{(n-1)!!},\frac{(n-1)!!}{(n-2)!!}\right\} \sim \sqrt{\frac {2 n} \pi}.
  \end{equation}
\end{theorem}
\begin{proof}
  We prove the theorem by induction on $n$. When $n = 1$, it is trivial. When $n =k> 1$, first we consider $\gamet(k,s,2)$. If Player 2 chooses $x_{k,1} = \cdots = x_{k,k} = s / k$, then the result will be \[\max\{\gamet(k -1, s-s/k,1), s/k\} = \max\left\{\frac{s(k-1)!!}{k!!}, s/k\right\} = \frac{s(k-1)!!}{k!!}. \] Player 1 can always choose the minimum $x_{k,i}$, which is no larger than $s / k$, so \[\gamet(k,s,2)\geq \gamet(k -1, s-s/k,1)= \frac{s(k-1)!!}{k!!}. \]
  Therefore, $\gamet(k,s,2)$ has value $\frac{s(k-1)!!}{k!!}$.

  Similarly, for $\gamet(k,s,1)$, if Player 1 chooses $x_{k,1} = s, x_{k,2} = \cdots = x_{k,k} = 0$, the result will be \[\min\{s, \gamet(k -1, s,2)\} = \min\left\{s, \frac{s(k-2)!!}{(k-1)!!}\right\} = \frac{s(k-2)!!}{(k-1)!!}.\]
  Also since Player 2 can always choose the minimum $x_{k,i}$, where $0\leq x_{k,i}\leq s/k$, we have that
  \begin{align*}
  \gamet(k,s,1)~&\leq \max\{\gamet(k -1, s,2), s/k\}\\
  ~&= \max\left\{\frac{s(k-2)!!}{(k-1)!!}, s/k\right\} = \frac{s(k-2)!!}{(k-1)!!}.
  \end{align*}
  Therefore, $\gamet(k,s,1)$ has value $\frac{s(k-2)!!}{(k-1)!!}$.

  To see the asymptotic equivalence given in Equation \eqref{eq:B}, observe that
  \[
  \min\left\{\frac{n!!}{(n-1)!!},\frac{(n-1)!!}{(n-2)!!}\right\}=\frac{(n'+1)!!}{n'!!} = \frac{(n'+1)\binom{n'}{\frac{n'}{2}}}{2^{n'}}
  \]
  for $n'=2\cdot \lfloor \tfrac {n-1} 2 \rfloor$.  Standard bounds on the central binomial coefficient also suffice to satisfy the (nonasymptotic) bound on this expression claimed by Theorem \ref{t.target}.
\end{proof}

\subsection{The geometric setting}
\label{s.geometric}
Recall that in the geometric setting, the state is modeled as a subset $X\subseteq \R^2$ topologically equivalent to an open disc, and the measure of a subset $C$ is now $\int_C \phi$. Let $S_T$ be the target subset.

Then, formally, the game can be expressed as the following recursive procedure, in which $T$ is allowed in general to be a finite union of topological open discs, and which for $k=n$, $T=X$ returns the maximum measure of the target over all single districts):

\medskip

\newcommand{\gameT}{\mathrm{\textsc{Game}}_{3}}
\begin{algorithmic}[1]
    \Procedure{Game$_3$}{$k,T,A$}\Comment{Player $A$ divides first}
    \State Player $A$ chooses $k$ districts: disjoint topological open discs of equal measure $d_1,\ldots,d_n$, such that the closure of their union is $T$
    \State Player $B$ chooses an integer $i\in [n]$, for $\{A,B\}=\{1,2\}$
    \State \textbf{return} $\max\left\{\gameT\left(k -1, \text{cl}\left(\bigcup_{j\in [n]/\{i\}} d_j\right),B\right), r(d_i)\right\}$
   \EndProcedure
\end{algorithmic}

\medskip

\noindent Here cl$(\cdot)$ denotes the closure of a set, and $r(d_i)$ denotes the measure of $d_i \cap S_T$.

Since $\gameT$ is much more complicated than $\gamet$, it is difficult to characterize its exact output. Our goal, therefore, is to bound $\gameT(n, X,c)$, that is, we need to find a good enough strategy for Player 2 such that for any strategy of Player 1, the measure of the target subset will be small in every district. Specifically, we will show that Player 2 can ensure that every district contains at most $\frac{2}{\sqrt{n}}$ of the target.
But before we introduce this strategy, we need to establish several lemmas.

\begin{lemma}\label{lemma:divide-district}
If a set $C$ in the plane is a topological open disc of measure $s$, and contains a measurable subset $A$ of measure $a$, then for any positive integer $k$, $C$ can be divided into two topological open discs $D$ and $C'$ ($D$, $C'$ are disjoint and the closure of their union is the closure of $C$) such that the measure of $D$ is $s/k$ and the measure of $D \cap A$ is $a/k$.
\end{lemma}

\begin{proof}
Without loss of generality we assume $C$ is an open disc centered at the origin. For any $\theta\in [0,2\pi]$, let $D(\theta)$ be the (open, say) circular sector of $C$ of measure $s/k$ starting from the angle $\theta$, let $f(\theta)$ be the measure of $D(\theta)\cap A$.  We can choose $\theta_1,\dots,\theta_k$ so that $D(\theta_1), \ldots, D(\theta_k)$ partition $C$ (up to boundaries) and among these circular sectors, we can find two circular sectors $D(\theta_i)$ and $D(\theta_j)$, such that $f(\theta_i) \leq a / k$ and $f(\theta_j) \geq a/ k$. Since $\phi$ is bounded, $f$ is a continuous function, and we can find an angle $\theta'$ between $\theta_i$ and $\theta_j$ such that $f(\theta') = a / k$. Thus we take $D=D(\theta')$ and $C'=C\setminus D$.
\end{proof}
By induction, Lemma \ref{lemma:divide-district} gives the following:
\begin{lemma}\label{lemma:divide-districts}
  If a set $C$ in the plane is a topological open disc of measure $s$, and contains a measurable subset $A$ of measure $a$, then for any positive integer $k$, $C$ can be divided into $k$ topological open discs $D_1, D_2, \ldots, D_k$ ($D_1, \ldots, D_k$ are disjoint and the closure of their union is the closure of $C$) such that for any $i \in [k]$, the measure of $D_i$ is $s/k$ and the measure of $D_i \cap A$ is $a/k$.
\end{lemma}

As play progresses from the initial configuration on $X$, players may choose districts that divide the remainder of the state into disconnected regions.  The following graph-theoretic lemma will allow us to control the effects of such strategies.

\begin{lemma}\label{lemma:split-graph}
Let $G=(V,E)$ be a connected graph with $n>1$ nodes. Each node $v\in V$ is labeled by a positive number $r(v)$. For any subgraph $C$, let $\bar r(C)$ be the average of $r$ over nodes in $C$. Moreover, let $\mathcal{C}_v$ be the set of connected components of the induced subgraph on $V\setminus \{v\}$. Then for any $c \geq 1$, there always exists a node $v$ such that $r(v) \leq \min\{c, n/2\}\bar r(G)$, and for any component $C \in \mathcal{C}_v$,
\begin{itemize}
  \item $c-1< |C| < n - c$ or $|C| = n - 1$,
  \item $\bar r(C) \leq \frac{|C|+1}{|C|}\bar r(G)$.
\end{itemize}
\end{lemma}

\begin{figure}[t]
\begin{center}

\subfigure[$\mathcal C_v = \{C_1,C_2,\ldots,C_m\}$.]{
\label{fig:split-graph}
\begin{tikzpicture}
\node[circle, draw] (v) at (0,0)  {$v$};
\foreach \x in {-1.75, -0.5, 1.75} {
  \draw (\x, -1) -- (v);
  \draw (\x, -1) -- (\x - 0.5, - 2.5);
  \draw (\x, -1) -- (\x + 0.5, - 2.5);
  \draw (\x - 0.5, -2.5) -- (\x + 0.5, - 2.5);
}
\node at (-1.75, -2) {$C_1$};
\node at (-0.5, -2) {$C_2$};
\node at (0.625, -2) {$\cdots$};
\node at (1.75, -2) {$C_m$};
\draw[rounded corners=10pt, dashed] (-2.5,-2.75) rectangle  (2.5, - 0.75);
\node at (-2.25, -0.5) {$\mathcal{C}_v$};
\end{tikzpicture}
}
\subfigure[$\mathcal C_{v_0} = \{C_1, C_2\}$; $v:\alpha$ means $r(v)=\alpha$.]{
\label{fig:split-graph-example}
\begin{tikzpicture}
\foreach \x/\y/\i/\r in {2.5/2/0/0, 0/1/1/0, 2.25/1/2/0, 5.25/1/3/0, 0/0/4/1, 1.5/0/5/0, 3/0/6/1, 4.5/0/7/1, 6/0/8/1, 1.5/-1/9/1}
  \node[draw] (\i) at (\x, \y) {$v_\i:\r$};
\foreach \i/\j in {0/1, 0/2, 0/3, 1/4, 2/5, 2/6, 3/7, 3/8, 5/9}
  \draw (\i) -- (\j);

\draw (4) -- (5);

\draw[rounded corners=10pt, dashed] (-0.67,-1.5) rectangle  (3.67, 1.5);
\node at (0.25, 1.75) {$\bar r(C_1) = 1/2$};
\draw[rounded corners=10pt, dashed] (3.83,-0.5) rectangle  (6.67, 1.5);
\node at (5.5, 1.75) {$\bar r(C_2) = 2/3$};
\end{tikzpicture}
}
\end{center}
\caption{Illustration of Lemma~\ref{lemma:split-graph}. (a) shows the structure of $\mathcal{C}_v$, and (b) gives an example of $G$, with $\bar r(G) = 1/2$. For $c=1$, $v_0,v_1,v_2,v_3,v_5$ all satisfy the conditions of the lemma. Taking $v_0$ as an example, $\mathcal C_{v_0} = \{C_1, C_2\}$, $|C_1|=6$, $|C_2|=3$,  $\bar r(C_1) = 1/2$ and $\bar r(C_2) = 2/3= (4/3)\bar r(G)$.}
\label{fig:example}
\end{figure}
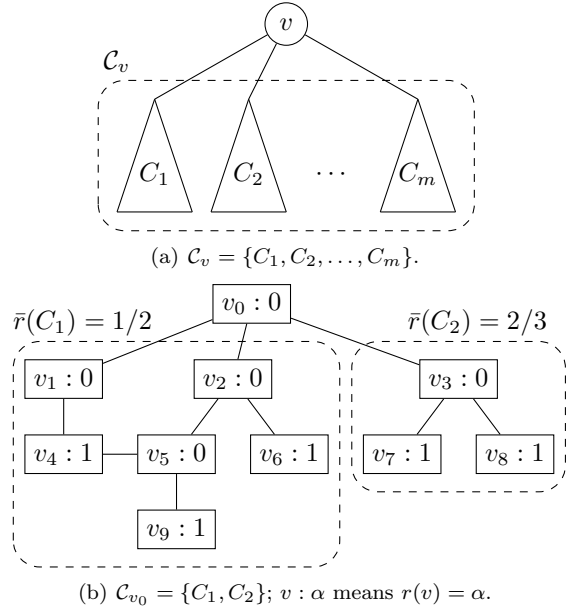

Figure \ref{fig:example} provides an illustrated example of the statement of Lemma \ref{lemma:split-graph} with 10 nodes. To prove the lemma, we will use the following simple observation:

\begin{lemma}\label{lemma:find-edge}
 Let $G=(V,E)$ be a connected graph, and let $T$ be a spanning tree of $G$. For an edge $e = (u,v)$, let the components of the induced subgraph of $T$ on $V\setminus \{v\}$  be $C_1(e), C_2(e),\ldots,C_{m(e)}(e)$. Without loss of generality, assume $u\in C_1(e)$, and let $C_0(e)$ denote  the subtree of $T$ spanned by $\{v\}\cup C_2(e)\cup\dots\cup C_m(e)$. Then there exists an edge $e$ such that $\bar r(C_0(e)) \leq \bar r(G)$ and for all $i=2,\ldots,m(e)$, $\bar r(C_i(e)) > \bar r(G)$.
\end{lemma}

The strong inequality at the end of the lemma's statement may seem confusing, as it could be the case that, say, $r(v)=1$ for all $v\in V$. But, in that case, we can choose $v$ to be a leaf, and then $C_1(e)$ is the rest of the tree, so $C_0(e)$ is just $v$ itself.

\begin{proof}[Proof of Lemma~\ref{lemma:find-edge}]
We find the edge $(u,v)$ with the following procedure. First, choose an arbitrary edge $e= (u,v)$. If $\bar r(C_0(e)) \leq \bar r(G)$, continue with $e$; otherwise, continue with $(v, u)$. (The desired inequality must hold for one of the two choices, as they induce the same partition of the vertices, with opposite choices of $C_0(e)$.) Then, while the current edge $e' = (u',v')$ has $2\leq i\leq m(e')$ such that $\bar r(C_i(e')) \leq \bar r(G)$, continue with $e''=(v', w)$ for $w\in C_i(e')$; otherwise, terminate with $(u,v)$. Notice that in the former case, $\bar{r}(C_0(e''))=\bar{r}(C_i(e'))\leq \bar{r}(G)$. See Figure \ref{fig:find-edge-example} for an illustration of the procedure.

This procedure will terminate because each edge can only be considered once. Moreover, the procedure ensures that for the current edge $e$, $r(C_0(e)) \leq \bar r(G)$. Finally, the termination condition ensures that for all $i=2,\ldots,m(e)$, $\bar r(C_i(e)) > \bar r(G)$.
\begin{figure}[p]
\begin{center}
\subfigure[A spanning tree of the graph shown in Figure \ref{fig:split-graph-example}.]{
\begin{tikzpicture}
\foreach \x/\y/\i/\r in {1.5/2/0/0, 0/1/1/0, 1.5/1/2/0, 3/1/3/0, 0/0/4/1, 1/0/5/0, 2/0/6/1, 2.75/0/7/1, 3.5/0/8/1, 1/-1/9/1}
  \node[circle, draw, inner sep=1pt] (\i) at (\x, \y) {$v_\i$};
\foreach \i/\j in {0/1, 0/2, 0/3, 1/4, 2/5, 2/6, 3/7, 3/8, 5/9}
  \draw (\i) -- (\j);

\draw[dashed] (4) -- (5);

\end{tikzpicture}
}
\hspace{0.8in}
\subfigure[First step: choose $e=(v_5,v_9)$.]{
\begin{tikzpicture}
\foreach \x/\y/\i/\r in {1.5/2/0/0, 0/1/1/0, 1.5/1/2/0, 3/1/3/0, 0/0/4/1, 1/0/5/0, 2/0/6/1, 2.75/0/7/1, 3.5/0/8/1, 1/-1/9/1}
  \node[circle, draw, inner sep=1pt] (\i) at (\x, \y) {$v_\i$};
\foreach \i/\j in {0/1, 0/2, 0/3, 1/4, 2/5, 2/6, 3/7, 3/8, 5/9}
  \draw (\i) -- (\j);

\draw[rounded corners=8pt, dashed] (0.6,-1.4) rectangle  (1.4, -0.6);
\node at (2.5, -1.2) {$\bar r(C_0(e)) = 1$};
\draw[rounded corners=10pt, dashed] (-0.4,-0.4) rectangle  (3.9, 2.4);
\node at (2.5, 2.6) {$\bar r(C_1(e)) = 4/9$};
\end{tikzpicture}
}
\vspace{0.5cm}

\subfigure[Second step: switch to $e=(v_9,v_5)$.]{
\begin{tikzpicture}
\foreach \x/\y/\i/\r in {1.5/2/0/0, 0/1/1/0, 1.5/1/2/0, 3/1/3/0, 0/0/4/1, 1/0/5/0, 2/0/6/1, 2.75/0/7/1, 3.5/0/8/1, 1/-1/9/1}
  \node[circle, draw, inner sep=1pt] (\i) at (\x, \y) {$v_\i$};
\foreach \i/\j in {0/1, 0/2, 0/3, 1/4, 2/5, 2/6, 3/7, 3/8, 5/9}
  \draw (\i) -- (\j);

\draw[rounded corners=8pt, dashed] (0.6,-1.4) rectangle  (1.4, -0.6);
\node at (2.5, -1.2) {$\bar r(C_1(e)) = 1$};
\draw[rounded corners=8pt, dashed] (-0.4,-0.4) -- (0.3, -0.4) -- (0.7, 0.7) -- (1.3, 0.7)  -- (1.7, -0.4)-- (3.9, - 0.4) -- (3.9, 1) --node[sloped, above]{$\bar r(C_2(e)) = 1/2$} (1.8, 2.4) -- (1.2, 2.4) -- (-0.4, 1.4) -- cycle;
\end{tikzpicture}
}
\hspace{0.8in}
\subfigure[Third step: switch to $e=(v_5,v_2)$.]{
\begin{tikzpicture}
\foreach \x/\y/\i/\r in {1.5/2/0/0, 0/1/1/0, 1.5/1/2/0, 3/1/3/0, 0/0/4/1, 1/0/5/0, 2/0/6/1, 2.75/0/7/1, 3.5/0/8/1, 1/-1/9/1}
  \node[circle, draw, inner sep=1pt] (\i) at (\x, \y) {$v_\i$};
\foreach \i/\j in {0/1, 0/2, 0/3, 1/4, 2/5, 2/6, 3/7, 3/8, 5/9}
  \draw (\i) -- (\j);

\draw[rounded corners=8pt, dashed] (0.6,-1.4) rectangle  (1.4, 0.4);
\node at (2.6, -1.2) {$\bar r(C_1(e)) = 1/2$};
\draw[rounded corners=8pt, dashed] (-0.4,-0.4) -- (0.4, -0.4) -- (0.4, 0.8) -- (1.5, 1.6)  -- (2.4, 0.8) -- (2.4, -0.4)-- (3.9, - 0.4) -- (3.9, 1) --node[sloped, above]{$\bar r(C_2(e)) = 1/2$} (1.8, 2.4) -- (1.2, 2.4) -- (-0.4, 1.4) -- cycle;
\draw[rounded corners=8pt, dashed] (1.67,-0.4) rectangle  (2.33, 0.4);
\node at (2.7, -0.7) {$\bar r(C_3(e)) = 1$};
\end{tikzpicture}
}
\vspace{.5cm}

\subfigure[Fourth step: switch to $e=(v_2,v_0)$.]{
\begin{tikzpicture}
\foreach \x/\y/\i/\r in {1.5/2/0/0, 0/1/1/0, 1.5/1/2/0, 3/1/3/0, 0/0/4/1, 1/0/5/0, 2/0/6/1, 2.75/0/7/1, 3.5/0/8/1, 1/-1/9/1}
  \node[circle, draw, inner sep=1pt] (\i) at (\x, \y) {$v_\i$};
\foreach \i/\j in {0/1, 0/2, 0/3, 1/4, 2/5, 2/6, 3/7, 3/8, 5/9}
  \draw (\i) -- (\j);

\draw[rounded corners=8pt, dashed] (0.6,-1.4) rectangle  (2.3, 1.4);
\node at (1.5, -1.7) {$\bar r(C_1(e)) = 1/2$};
\draw[rounded corners=8pt, dashed] (-0.4,-0.4) --node[sloped, above]{$\bar r(C_2(e)) = 1/2$} (-0.4, 1.4) -- (0.4, 1.4) -- (0.4,-0.4) --cycle;
\draw[rounded corners=8pt, dashed] (2.45, -0.4)-- (2.45, 1.4)  --(3.9, 1.4) -- (3.9, -0.4) --cycle;
\node at (3.7, -0.7) {$\bar r(C_3(e)) = 2/3$};

\end{tikzpicture}
}
\hspace{0.4in}
\subfigure[Fifth step: switch to $e=(v_0,v_1)$.]{
\begin{tikzpicture}
\foreach \x/\y/\i/\r in {1.5/2/0/0, 0/1/1/0, 1.5/1/2/0, 3/1/3/0, 0/0/4/1, 1/0/5/0, 2/0/6/1, 2.75/0/7/1, 3.5/0/8/1, 1/-1/9/1}
  \node[circle, draw, inner sep=1pt] (\i) at (\x, \y) {$v_\i$};
\foreach \i/\j in {0/1, 0/2, 0/3, 1/4, 2/5, 2/6, 3/7, 3/8, 5/9}
  \draw (\i) -- (\j);

\draw[rounded corners=8pt, dashed] (0.6,-1.4) rectangle  (3.9, 2.4);
\node at (2, -1.7) {$\bar r(C_1(e)) = 1/2$};
\draw[rounded corners=8pt, dashed]  (-0.4,-0.4) --node[sloped, above]{$\bar r(C_2(e)) = 1$} (-0.4, 0.4) -- (0.4, 0.4) -- (0.4,-0.4) --cycle;
\end{tikzpicture}
}
\caption{The procedure described in Lemma~\ref{lemma:find-edge}, applied to the graph of Figure~\ref{fig:split-graph-example}.}
\label{fig:find-edge-example}
\end{center}
\end{figure}
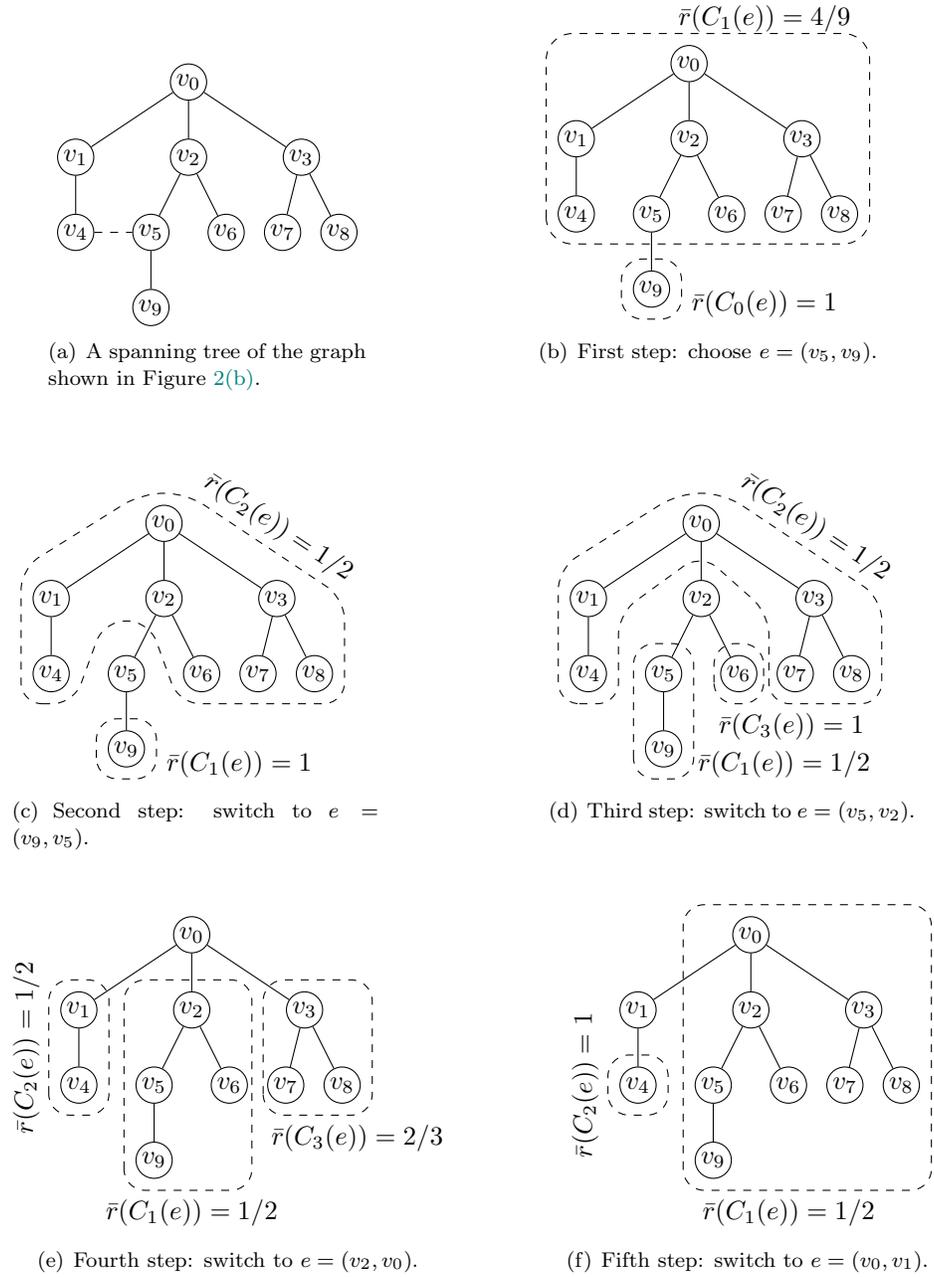
\end{proof}

\begin{proof}[Proof of Lemma \ref{lemma:split-graph}]

  First, if there is a leaf $v$ of $T$ with \[r(v) \leq \min\{c, n/2\}\bar r(G),\] then $v$ trivially satisfies the conditions. Moreover, if $c \geq n / 2$, there must be at least one leaf satisfying the foregoing inequality, because there are at least two leaves in any tree. Therefore, in what follows, we assume that $c < n / 2$, and every leaf $v$ of $T$ satisfies
\begin{equation}
\label{eq:rv}
r(v) > c\bar r(G).
\end{equation}
Now let the edge whose existence is guaranteed by Lemma \ref{lemma:find-edge} be $e = (u,v)$; we will show that the node $v$ satisfies the conditions of Lemma \ref{lemma:split-graph}.

  For $S\subseteq [m(e)]$, let $C_S(e) = \bigcup_{i\in S} C_i(e)$. We claim that for all $S \subseteq [m(e)]$, we have 
  \begin{equation}
    \label{e.rbound}\bar r(C_S(e) \cup \{v\})\leq \bar r(G).
  \end{equation}
  To see this, suppose it fails.  If $1\notin S$, then we get the contradiction $\bar r(C_0(e)) > \bar r(G)$, since $C_0(e) = (C_S(e) \cup \{v\})\cup \bigcup_{2\leq j\leq m(e),j\notin S} C_j(e)$ and $\bar r(C_j(e))>\bar r(G)$ for $2\leq j\leq m(e)$ by assumption.  If $1\in S$, then we get the analogous contradiction  $\bar r(G) > \bar r(G)$, since in this case $G = (C_S(e) \cup \{v\})\cup \bigcup_{2\leq j\leq m(e),j\notin S} C_j(e)$.

  Note that taking $S = \emptyset$ gives the Lemma's claim that $r(v) \leq \bar r(G)$. Moreover, since the $C_i$'s are connected subgraphs of of the induced subgraph on $V\setminus \{v\}$ that partition its vertex set, any component $C\in \mathcal{C}_v$ must have the same vertex set as $C_S(e)$ for some $S\subseteq [m(e)]$. Therefore, \eqref{e.rbound} gives that  all components $C\in \mathcal{C}_v$ satisfy $\bar r(C \cup \{v\})\leq \bar r(G)$, or, since $r(v)\geq 0$, that
  \begin{equation*}
    \bar r(C)\leq \frac{|C|+1}{|C|}\bar r(G).
  \end{equation*}
  This is the second requirement the lemma places on components $C\in \mathcal{C}_v$.

  To establish the first, consider a nonempty $S\subseteq [m(e)]$. Observe that $C_S(e) \cup \{v\}$ contains at least one leaf of $T$. Using \eqref{eq:rv}, we have
  \begin{equation}
    \label{e.fromoneleaf}
    \bar r(C_S(e) \cup \{v\}) > \frac{c\bar r(G)}{|C_S(e)| + 1}.
  \end{equation}
 Now \eqref{e.rbound} and \eqref{e.fromoneleaf} imply that
  \begin{equation}
    \label{e.CSbound}
    |C_S(e)|+1 > c.
  \end{equation}
Moreover, if $|C| \neq n - 1$, that is, $ [m(e)]\setminus S \neq \emptyset$, then \eqref{e.CSbound} gives that
  \begin{equation}\label{e.sizerange}
    |C| = n - |C_{[m(e)]\setminus S}(e)|- 1 < n - c,
  \end{equation}
  completing the proof.
\end{proof}

The significant source of complication in \textsc{Game}$_3$ is that the state may be disconnected during play.  In order to illustrate our strategy clearly, we will need to classify the components of the active part of the state into 3 types. At the beginning, if the size/measure of the state is odd, the initial component (the whole state) is Type 1; otherwise, it is Type 2. If a component (of any type) is split into several smaller components, these components are also assigned Type 1 or Type 2 according to the parity of their size.  If a player freezes a district in a component, but keeps the component connected, then the type of the component changes as follows:

\begin{itemize}
  \item Type 1 becomes Type 2;
  \item Type 2 becomes Type 1 if chosen by Player 1, and becomes Type 3 if chosen by Player 2;
  \item Type 3 becomes Type 2.
\end{itemize}
Note that at any point, a component with an even number of districts is Type 2, while a component with an odd number of districts is Type 1 or Type 3.
Our strategy for Player 2 is as follows.

\begin{itemize}
  \item If it is Player 2's turn to cut, for each connected component $C$, there is a topological open disc $C'$ whose closure is $C$. Then by applying Lemma \ref{lemma:divide-districts}, he divides it into districts such that the proportion of the target is the same in every district.
  \item If it is Player 2's turn to freeze, he plays in a component of Type 1 if available, and otherwise in component of Type 2.  (We will prove that one of these types is always available to him.)   To choose the district to freeze, Player 2 regards the presented districting of his selected component as a planar graph, with districts as vertices and edges corresponding to shared boundaries of positive length.  He applies Lemma \ref{lemma:split-graph} to this graph with $c=3$ and freezes the district corresponding to the vertex $v$ given by the lemma.
\end{itemize}

Depending on which player freezes first and the parity of the measure of the original state, a given districting game is either of \emph{odd type}, where Player 2 is always presented with an odd number of districts when it is his turn to freeze, of of \emph{even type}, where he encounters an even number of districts when it is his turn to freeze.
\begin{lemma}\label{lemma-no-3}
  In an odd type instance of $\gameT$, there will never exist component of Type 3 if Player 2 plays as above.
\end{lemma}
\begin{proof}
This is a simple consequence of the parities and Player 2's preference to play in Type 1 components.   Note that there are initially no Type 3 components and only Player 2's freezes can create such components.  In particular, we prove by induction that immediately before Player 2's $k$th turn to freeze a district, there are no Type 3 components.  Note that since the game is of odd type, this induction hypothesis implies that there is at least one Type 1 component before his $k$th turn to freeze.  Player 2's strategy will thus freeze a district in a Type 1 component, and no Type 3 components will be created on his $k$th turn, meaning that immediately before his $(k+1)$st turn to freeze, no Type 3 components will exist, as desired.
\end{proof}

\begin{lemma}
  At any point, there exists at most one component of Type 3. If it is Player 2's turn to freeze, there exists at least one component of Type 1 or 2.
\end{lemma}
\begin{proof}
  Note that the first claim of the Lemma implies the second: Indeed, if it is Player 2's turn to freeze, and the only component available is of Type 3, then the game is of odd type, and this contradicts Lemma \ref{lemma-no-3}.

  To prove the first claim, we proceed again by induction. Suppose a choice by Player 2 increased the number of Type 3 components; in this case, there were no Type 1 components available, and by Lemma \ref{lemma-no-3}, the game is of even type; this implies the number of Type 3 components available was even. Since it is at most 1 by assumption, there were no Type 3 components available, and thus the number of Type 3 components can only increase to 1, as claimed.
\end{proof}

For each connected component $C$, define $r(C)$ as the ratio of the target in $C$, i.e., \[ r(C)= \frac{\int_{C\cap S_T} \phi}{\int_C \phi}.\] The following lemma establishes an upper bound on $r(C)$ according to its type.

\begin{lemma}\label{lemma-bound-by-type}
  Let $r_0 = r(X)$, where $X$ is the initial state. For any sequence $(a_1, a_1, a_2, \ldots, a_k)$, let $q(a_1,a_2,\ldots, a_k) = \prod_{i=1}^k \frac{a_i}{a_i-1}$. At any step of the game, for any connected component $C$ of measure $s$:
  \begin{itemize}
    \item if $C$ is Type 1, \[r(C)\leq r_0\cdot q(s+2,s+4,\ldots, n'');\]
    \item if $C$ is Type 2, \[r(C)\leq r_0\cdot q(s+1,s+3,\ldots, n'');\]
    \item if $C$ is Type 3, \[r(C)\leq r_0\cdot q(s+1,s+2,s+4,\ldots, n''),\]
  \end{itemize}
  where $n'' = 2\lfloor (n-1)/2\rfloor + 1$.
\end{lemma}

\begin{proof}
  We prove the lemma by induction on $s$. At the beginning of the game, it is trivially true. Now suppose a component $C$ is chosen, there are two cases, according to whether the component is split by frozen district.

  First, we consider the case where the component is still connected after the move. After removing a district from $C$, let the new component be $C'$ and suppose it contains $s$ districts. If it is Player 1's turn to freeze, then $r(C') = r(C)$; otherwise, $r(C') \leq \frac{s+1}{s}r(C) = q(s+1)r(C)$. For any possible type change, one can easily verify that $r(C')$ satisfies the inequality; here we only take changes from Type 2 to Type 3 as an example:
  \begin{align*}
   r(C') ~&\leq q(s+1)r(C) \leq q(s+1) \cdot r_0\cdot q(s+2,s+4,\ldots,n'') \\
         ~&=  r_0\cdot q(s+1,s+2,s+4,\ldots,n''),
  \end{align*}
where the second inequality holds by the induction assumption.

  Second, we consider the case where smaller components emerge. Let one of the new components be $C'$ and suppose it contains $s$ districts. If it is Player 1's turn to freeze, then we have $r(C') = r(C)$ and $|C'|\leq |C|-2$. In this case, one can easily verify that $r(C')$ satisfies the inequalities.

If it is Player 2's turn to freeze, then we have $r(C') \leq \frac{s+1}{s} r(C)$, and $C$ contains more than $s+c$ districts by Lemma \ref{lemma:split-graph} (as $s=|C'|<|C|-c$). This case is more complicated, and we have to enumerate types of $C'$ and $C$ -- both can be Type 1 or 2. Let us first take $C'$ of Type 1 and $C$ of Type 2 as an example. Suppose $C$ contains $s+2a+1$ districts, then
  \[r(C) \leq r_0 \cdot q(s+2a+2,s+2a + 4, \ldots, n'').\]
  In order to ensure \[q(s +1) r(C) \leq r_0\cdot q(s + 2, s + 4, \ldots, n''),\]
  we only need to make sure \[q(s+1)\leq q(s+2,s+4,\ldots, s+2a).\] Since $c = 3$ and $c<s<s+2a+1-c$, we have $a \geq 2$ and $s \geq 4$, which satisfies the inequality above.
  Similarly, one can easily verify that when both $C'$ and $C$ are Type 1, $c\geq 2$ is enough; when $C'$ is Type 2, $c\geq 1$ is enough. Since we let $c = 3$ when applying Lemma \ref{lemma:split-graph}, $r(C')$ satisfies the inequalities.
\end{proof}

\begin{proof}[Proof for Theorem \ref{t.target}]
  For any component $C$ of any type, by Lemma \ref{lemma-bound-by-type},
  \[r(C)\leq\left\{\begin{array}{ll} r_0 \cdot \frac{n''!!(s-2)!!}{(n''-1)!!(s-1)!!}, & \text{if } C \text{ is Type 2}, \\
  r_0\cdot\frac{s+1}{s}\cdot \frac{n''!!(s-1)!!}{(n''-1)!!s!!}, & \text{otherwise}. \end{array}\right.\]
  If Player 1 freezes a district $d$ from $C$, $r(d) = r(C)$. If Player 2 freezes a district $d$ from $C$, by Lemma \ref{lemma:split-graph} (with $c = 3$), $r(d)\leq \min\{3, s/2\}r(C)$ if $s > 1$, otherwise $r(d) = r(C)$. Therefore, for any $d$,
  \begin{align*}
    r(d)\leq&~ r_0 \cdot \frac{n''!!}{(n''-1)!!} \cdot \max\left\{2, \max_{k\in \mathbb{N}^+}\left\{ \min\{3, (2k - 1)/2\} \frac{2k}{2k-1} \frac{(2k-2)!!}{(2k-1)!!}\right\} , \right.\\
    &~ \left.\max_{k\in \mathbb{N}^+}\left\{\min\{3, k\} \frac{(2k-2)!!}{(2k-1)!!}\right\}\right\}\\
    =&~ r_0 \cdot \frac{n''!!}{(n''-1)!!} \cdot \max\left\{2, \frac{3 \cdot 6!!}{7!!}, \frac{3 \cdot 4!!}{5!!}, \frac{2 \cdot 2!!}{3!!}\right\} \\
    =&~ 2r_0 \cdot \frac{n''!!}{(n''-1)!!}.
  \end{align*}

  Since \[\frac{n''}{n''-1} \leq \sqrt \frac{n''}{n''-2},\]
  by induction on $n''$, one can easily find
  \[\frac{n''!!}{(n''-1)!!} \leq \sqrt{n''}\leq \sqrt{n}.\]
  That is,
  \[r(d)\leq 2r_0\sqrt{n}.\]

  In other words, $d$ contains at most $\frac{2}{\sqrt{n}}$ of the target.
\end{proof}

\section{Limitations and further questions}
\label{s.realworld}

We have analyzed optimum play of the I-cut-you-freeze protocol in idealized settings; in actual applications to redistricting, real-world constraints would interfere with optimum play. For example, our analysis considers a district with $50.1\%$ loyalty to Player $A$ to be controlled by that player, whereas in a real-world setting, a Player would want a more comfortable loyalty margin to consider a district safe.  Apart from this thresholding issue, two  other simplifications in our model stand out:
\begin{itemize}
\item \textbf{Geometric constraints on districts.} In the United States, Congressional districts are required to be connected, but in many states, are also required to be geometrically ``nice'' in other less-precise ways.  One common term used in state-by-state requirements on districts is that they be ``compact'', which is supposed to limit the extent to which districts have intricate drawn out structure.  Various metrics have been proposed to quantify the ``compactness'' of a district \cite{horn1993practical}; one of the simplest is the ratio $4\pi A_D/P_D^2$, where here $A_D$ and $P_D$ are the area and perimeter of the district, respectively.  Note that with this normalization, the measure takes a value in $(0,1]$.  Other common requirements include respect for geographical phenomena such as cities, counties, etc.

  Our analysis necessarily ignores these geometric constraints on the districts.  (It should be noted that there are no precise and agreed upon definitions of what constitutes a valid district in any particular state, and in practice, many Congressional districts seem to flaunt natural interpretations of these constraints.)

\item \textbf{Mixed populations.}  Theorem \ref{t.seats} concerns a model of redistricting in which districts can be assembled with essentially arbitrary collections of voters among those voters remaining in the unfrozen part of the state.   In practice, however, Democrats and Republicans are sometimes neighbors, and it is generally impossible to draw a district which is $100\%$ loyal to party $A$ or party $B$, as is sometimes an optimum move for our protocol (i.e., as from Lemma \ref{lemma:weaker-strategy}).
\end{itemize}
It may be of theoretical interest to analyze our protocol in richer models motivated by these complications.

  That said, we believe it is reasonable to infer basic real-world properties of our protocol from our rigorous analysis in idealized settings.  Let us first consider Theorems \ref{t.a-seats} and \ref{t.seats} concerning the slate which will result from optimum play in our algorithm.  For these results, it is reasonable to suspect that the idealized model we work in has a significant effect on the precise results we obtain.  In particular, our proof is based on the feasibility of two types of moves for the players: Lemma \ref{lemma:stronger-strategy} requires a player to divide a region into districts with similar proportions of voters loyal to each party, while Lemma \ref{lemma:weaker-strategy} also requires him to draw some districts with pure loyalty for one party.   Obviously, neither of these moves is perfectly possible to emulate in a real-world situation; both players will be handicapped to some degree.

  Nevertheless, we consider the key conclusion of Theorems \ref{t.a-seats} and \ref{t.seats} to be not the particular formula for the slate won by each player in our protocol, but instead the general feature that our protocol does shift the unbalanced seat/loyalty curve from the ``One-player-decides'' protocol to an (asymptotically) symmetric curve (recall Figure~\ref{f.curve}).  In some sense, the key point of Theorems \ref{t.a-seats} and \ref{t.seats} is just that the protocol produces a result within reason, and that neither player gains a significant advantage from the choice of who is assigned the first move in the protocol; we expect that both of these properties would persist in real-world applications.

  For Theorem \ref{t.target} we believe there is actually relatively little lost in our abstraction of the real-world problem.  Much of our analysis (e.g., Lemma \ref{lemma:split-graph}) concerns the case where Player 2 is choosing which district to freeze after Player 1 has divided the state; this part of our analysis holds regardless of what geometric constraints are imposed on the divisions made by Player 1 on his turn.  When Player 2 must divide a region, our analysis directs him to divide the target evenly among many districts.  In practice, divisions with the general feature of dividing a target among many districts are generally quite easy to construct, and thus we expect that our protocol would retain a property similar in spirit to the assertion of Theorem \ref{t.target}.  The main takeaway from Theorem \ref{t.target} is not the precise threshold of $B$ for which the protocol has the $B$-target property, but the fact that neither player will be able to completely direct the composition of any particular district; to some approximation, we expect this property to survive in real-world implementations.

\subsection*{Acknowledgments}  The first two authors are partially supported by the National Science Foundation and the Sloan Foundation; the second author is additionally supported by the Office of Naval Research.  The first author also acknowledges helpful discussions with Maria Chikina on the problem setting considered here.

\bibliographystyle{plain}

\begin{thebibliography}{10}

\bibitem{duke}
Sachet Bangia, Christy~Vaughn Graves, Gregory Herschlag, Han~Sung Kang, Justin
  Luo, Jonathan~C Mattingly, and Robert Ravier.
\newblock Redistricting: Drawing the line.
\newblock {\em arXiv preprint arXiv:1704.03360}, 2017.

\bibitem{BT96}
Steven~J Brams and Alan~D Taylor.
\newblock {\em Fair Division: From Cake-Cutting to Dispute Resolution}.
\newblock Cambridge University Press, 1996.

\bibitem{carson2014reevaluating}
Jamie~L Carson, Michael~H Crespin, and Ryan~D Williamson.
\newblock Reevaluating the effects of redistricting on electoral competition,
  1972--2012.
\newblock {\em State Politics \& Policy Quarterly}, 14(2):165--177, 2014.

\bibitem{chenrodden}
Jowei Chen, Jonathan Rodden, et~al.
\newblock Unintentional gerrymandering: Political geography and electoral bias
  in legislatures.
\newblock {\em Quarterly Journal of Political Science}, 8(3):239--269, 2013.

\bibitem{outliers}
Maria Chikina, Alan Frieze, and Wesley Pegden.
\newblock Assessing significance in a {M}arkov chain without mixing.
\newblock {\em Proceedings of the National Academy of Sciences},
  114(11):2860--2864, 2017.

\bibitem{horn1993practical}
David~L Horn, Charles~R Hampton, and Anthony~J Vandenberg.
\newblock Practical application of district compactness.
\newblock {\em Political Geography}, 12(2):103--120, 1993.

\bibitem{LRY}
Zeph Landau, Oneil Reid, and Ilona Yershov.
\newblock A fair division solution to the problem of redistricting.
\newblock {\em Social Choice and Welfare}, 32(3):479--492, 2009.

\bibitem{LandauSu}
Zeph Landau and Francis~Edward Su.
\newblock Fair division and redistricting.
\newblock {\em AMS Special Sessions on The Mathematics of Decisions, Elections,
  and Games}, pages 17--36, 2010.

\bibitem{lindgren2013effect}
Eric Lindgren and Priscilla Southwell.
\newblock The effect of redistricting commissions on electoral competitiveness
  in {US} house elections, 2002--2010.
\newblock {\em Journal of Politics and Law}, 6(2):13--18, 2013.

\bibitem{mcdonald}
Michael~D McDonald and Robin~E Best.
\newblock Unfair partisan gerrymanders in politics and law: A diagnostic
  applied to six cases.
\newblock {\em Election Law Journal}, 14(4):312--330, 2015.

\bibitem{Pro13}
Ariel~D Procaccia.
\newblock Cake cutting: Not just child's play.
\newblock {\em Communications of the ACM}, 56(7):78--87, 2013.

\bibitem{wang}
Samuel S-H Wang.
\newblock Three tests for practical evaluation of partisan gerrymandering.
\newblock {\em Stanford Law Review}, 68(6):1263, 2016.

\end{thebibliography}

\end{document}